\documentclass[12pt]{article}

\usepackage[margin=1in]{geometry}
\usepackage{amsmath,amssymb,amsthm,mathtools,bm}
\usepackage{booktabs}
\usepackage{array}
\usepackage{microtype}
\usepackage[backend=biber,style=authoryear,maxcitenames=2,maxbibnames=99,natbib=true]{biblatex}
\usepackage{hyperref}
\hypersetup{hidelinks}
\usepackage{graphicx}
\usepackage{xcolor}
\usepackage{placeins}
\usepackage{setspace}
\newtheorem{theorem}{Theorem}

\newcommand{\tr}{\operatorname{tr}}

\newcommand{\norm}[1]{\left\lVert #1\right\rVert}

\title{\bfseries Longitudinal Brain Aging in Covariance Space:\\Temporally Calibrated Bayesian Inference for Adaptive Covariance Operators}
\author{
Montserrat Fuentes, PhD\\
St. Edward's University
\and
Yi Zhao, PhD\\
Department of Biostatistics and Health Data Science\\
Indiana University
}
\date{}

\begin{document}
\maketitle

\begin{center}
\end{center}
\vspace{0.5em}

\begin{abstract}
Functional brain aging is usually described from differences between people of different ages, even though aging itself unfolds within the same individual. This distinction becomes especially important for functional connectivity, where the scientific object is a high-dimensional covariance system and the direction of reproducible change may itself evolve across adulthood. We introduce a novel Bayesian longitudinal covariance-operator framework that treats each visit covariance as a latent positive-definite state informed directly by repeated resting-state fMRI runs. The proposed approach combines several innovations: run-level information is calibrated for temporal dependence and replicate variability, within-person aging is separated from between-person age differences, and an adaptive tangent representation learns the low-dimensional structure of systematic longitudinal change without selecting directions from pooled covariance variation. The longitudinal operator is allowed to change smoothly in magnitude and orientation with age, and exact Fr\'echet derivatives translate the fitted trajectory back to the covariance scale. This construction makes it possible to study questions that are not addressed by standard cross-sectional analyses or fixed-projection longitudinal approaches: how the covariance organization of the same brain changes with age, when that reorganization is most pronounced, whether it follows the population age gradient, and which functional systems account for agreement or divergence. In the Dallas Lifespan Brain Study, within-person covariance aging differs markedly from the cross-sectional age gradient through early and middle adulthood, with the discrepancy weakening later in life. Reorganization is strongest around midlife, the systematic longitudinal signal is concentrated in only a few coordinated directions, and most of its off-diagonal energy is expressed between functional systems. The results show that population-level signatures of brain aging do not necessarily reflect the trajectory followed by an individual brain and demonstrate the value of modeling longitudinal covariance change directly as a biological process.
\end{abstract}

\section{Introduction}
\label{sec:intro}

Healthy brain aging is increasingly understood as a reorganization of interacting neural systems rather than a uniform loss of connectivity. Resting-state functional magnetic resonance imaging (rs-fMRI) has made this organization observable through spontaneous correlations among distributed brain regions. Across adulthood, changes have been documented in long-range connectivity, network integration, modular structure, and the balance between within-system and between-system coupling \citep{damoiseaux2008,betzel2014,salallonch2014,geerligs2015}. More recent studies continue to show that aging is associated with changes in the organization and differentiation of large-scale functional systems rather than a simple global reduction in connectivity \citep{setton2023age,martin2023age,han2024measures,calder2026high}. This broader view is consistent with neurocognitive theories in which adaptation and compensatory reorganization accompany decline \citep{parkreuter2009,reuterpark2014}. The relevant scientific object is therefore not a collection of isolated connections, but the organization of dependence across the brain as a whole.

The Dallas Lifespan Brain Study (DLBS) has played an important role in shaping this perspective. Cross-sectional analyses showed that younger adults exhibit stronger functional-system segregation, whereas increasing age is associated with less differentiated large-scale organization \citep{chan2014,chan2017}. These findings, together with more recent evidence of age-related reorganization across functional systems \citep{setton2023age,martin2023age,han2024measures,calder2026high}, make the distinction between population age differences and individual aging increasingly important. A pattern that distinguishes older from younger adults is not necessarily the path followed by the same brain as it ages. Population age contrasts combine longitudinal change with stable differences among people and with cohort structure. In a high-dimensional covariance system, this distinction can involve more than a difference in effect size: the direction of covariance change within a person may differ from the direction that separates people of different ages. A cross-sectional signature may therefore characterize how older and younger individuals differ without describing how an individual brain changes over time.

Longitudinal rs-fMRI studies provide evidence that functional organization changes within individuals, including changes in connectivity, segregation, modularity, and network configuration \citep{ng2016,oschmann2020,malagurski2020}. Yet most longitudinal analyses reduce each scan to selected edges, predefined network summaries, graph measures, or a small number of connectivity features. Those summaries are valuable for focused questions, but they do not describe the evolution of the covariance state itself. They cannot determine whether the dominant covariance direction of within-person aging is stable across adulthood, whether the rate of whole-system reorganization changes with age, or whether the longitudinal trajectory agrees geometrically with the population age gradient.

Statistical methods for covariance outcomes provide important foundations for addressing these questions. Matrix-logarithmic and geometric approaches preserve positive-definite structure \citep{chiu1996,yuan2008}, while covariance regression and Bayesian factor models allow covariates to alter dependence structure \citep{hoffniu2012,foxdunson2015}. Covariate Assisted Principal regression introduced a way to learn covariance directions associated with predictors rather than selecting them from pooled variation alone \citep{zhao2021}, and its longitudinal extension accommodates repeated covariance outcomes through a learned projection \citep{zhao2024}. Recent developments have broadened this literature further. Bayesian CAP provides posterior inference for covariance-associated directions \citep{park2025bayesian}, while multilevel covariance regression introduces hierarchical structure for repeated and clustered covariance outcomes \citep{orindi2023combined}. MCAP allows covariance-associated projection directions to vary across groups while sharing information across them, with a cross-sectional lifespan application that identifies population subnetworks that differ across age \citep{green2026}. Latent-space formulations provide another strategy for representing high-dimensional covariance structure parsimoniously \citep{wang2026latent}.

These methods provide a foundation for studying covariance variation, but our focus is the continuous trajectory of the covariance within individuals. Unlike approaches based on a fixed covariance-associated projection, we allow the direction of longitudinal change to evolve with age. The geometric and neural-field modes studied by Robinson and Pang \citep{robinson2016eigenmodes,pang2023geometric} arise from anatomical or dynamical operators. Our covariance modes instead describe dependence among BOLD signals and are estimated from repeated measurements.

The visit-specific covariance is inferred from finite BOLD time series rather than observed directly. Filtering, temporal dependence, and censoring affect the information available from each run. Repeated runs also vary around the underlying visit covariance. Our model carries this uncertainty into inference on longitudinal change rather than treating an estimated covariance matrix as error-free.

High dimensionality creates a related modeling challenge. With 100 cortical regions, a visit covariance contains 5,050 distinct entries. Regularization is unavoidable, but the source of the low-dimensional structure matters. Directions that dominate pooled covariance variation may reflect stable individual heterogeneity, overall amplitude, or nuisance variation rather than reproducible aging. We therefore place adaptive dimension reduction on the systematic within-person aging component while allowing baseline covariance and individual heterogeneity to remain substantially richer. The learned subspace is tied to the longitudinal scientific question rather than inherited from an unsupervised decomposition of total variation.

A Bayesian hierarchy links repeated runs to a latent positive-definite covariance state at each visit. It then learns the longitudinal trajectory and the uncertainty in its direction, magnitude, and effective dimension. Age is decomposed into within-person change around each participant's mean observed age and between-person location in the population age distribution. Separate operators represent these two sources of age variation, so their similarity is learned rather than assumed. The within-person operator varies smoothly with age inside an adaptive tangent span, allowing the covariance-aging direction to reorient across adulthood. Exact Fr\'echet derivatives then translate the tangent-space model to covariance scale, where the scientific questions are naturally expressed.

The fitted trajectory describes the rate and direction of within-person covariance change and its agreement with the cross-sectional age gradient. Network-resolved summaries show which systems contribute to this agreement or divergence without requiring separate tests for thousands of edges.

The DLBS provides a particularly informative setting because the study has helped establish cross-sectional functional segregation as a hallmark of healthy aging and now contains repeated observations spanning much of adulthood. The longitudinal data allow that influential population-level picture to be reconsidered as a within-person biological process. Agreement between the longitudinal covariance trajectory and the cross-sectional gradient would support the interpretation that population age signatures trace individual aging. Divergence would show that the two describe different aspects of lifespan variation and would identify when and where that distinction is strongest.

The remainder of the paper follows the analysis from imaging data to biological interpretation. Section~\ref{sec:data} introduces the DLBS observations and defines the scientific objective. Section~\ref{sec:model} describes how repeated runs inform a longitudinal covariance model that separates changes within individuals from differences across individuals. Section~\ref{sec:method} explains how the fitted trajectory is summarized in biologically interpretable measures. Section~\ref{sec:inference} presents the Bayesian approach and the positive-definiteness property. Section~\ref{sec:dlbs} reports the DLBS findings, and Section~\ref{sec:discussion} considers their implications. Mathematical details, the proof, and simulation results are provided in the Supplementary Materials.

\section{Data and scientific target}
\label{sec:data}

The Dallas Lifespan Brain Study was designed to characterize neural and cognitive change across adulthood. The full study enrolled 464 adults aged 21--89 years at the initial assessment. Follow-up waves were conducted approximately every 3.5--5 years, with 338 participants returning for a second epoch and 224 for a third \citep{park2025}. The broad baseline age distribution is important for the present analysis because repeated measurements begin at many points in adulthood rather than being concentrated only in late life.

Our methodological analysis uses 50 participants selected to provide broad lifespan coverage while retaining three repeated visits. The resulting data contain 150 participant-visits and 252 resting-state runs. Across all observed visits, ages range from 22 to 96 years. The cohort contains 28 women and 22 men. Most participants contribute five runs across the three visits, while two contribute six. The resting-state acquisition used a repetition time (TR) of 2 s throughout the study. Epoch 1 included a 154-volume resting-state run, whereas epochs 2 and 3 included two runs per visit; the run duration was increased from 154 to 180 volumes partway through epoch 2 \citep{park2025}. Thus, the difference between 154 and 180 original volumes reflects a change in run duration rather than temporal sampling.
After preprocessing and framewise-displacement censoring, retained run length ranges from 89 to 180 volumes, with a mean of 156.

All BOLD images were preprocessed with fMRIPrep \citep{esteban2019fmriprep} and analyzed in the \texttt{MNI152NLin2009cAsym} template space. Regional signals were extracted with the Schaefer 2018 cortical atlas \citep{schaefer2018} using 100 parcels grouped into seven canonical networks. The atlas was obtained directly in the same nonlinear template space and at 2-mm resolution, avoiding cross-template interpolation of the label image. Regional series were linearly detrended and band-pass filtered between 0.01 and 0.10 Hz. Nuisance regression included the six rigid-body motion parameters together with their temporal derivatives, squares, and squared derivatives, plus white-matter and cerebrospinal-fluid signals. Volumes with framewise displacement above 0.5 mm were censored. Runs remained separate throughout the analysis rather than being concatenated within visit.

The primary representation uses percent signal change (PSC), rather than standardizing each parcel to unit variance. This choice is important for the present scientific target because unit-variance standardization removes regional amplitude information by construction and produces a correlation-like representation. PSC preserves relative BOLD fluctuation amplitude, allowing the covariance model to distinguish changes in marginal variability from changes in dependence structure. A z-standardized representation is retained as a prespecified sensitivity analysis to assess whether the main conclusions depend on amplitude information.

Quality-control diagnostics indicated that residual run-level motion burden remained associated with features of the estimated covariance matrices after preprocessing and censoring. We therefore include run-level motion as a nuisance covariate in the longitudinal covariance model to reduce the possibility that residual motion-related variation is attributed to age. This model-based adjustment is complementary to preprocessing and does not reintroduce motion into the BOLD time series. Robustness was further assessed using stricter censoring and exclusion of high-motion observations.

Figure~\ref{fig:dlbs_brain} connects the imaging representation to the statistical object estimated by the model. Panel A shows the Schaefer-100 cortical parcellation organized into seven functional networks. Each visit therefore contributes repeated multivariate BOLD observations on the same 100-region cortical system. The goal is not to estimate pairwise covariance changes independently, but to characterize how the covariance system changes as a coordinated whole.

Biologically, the covariance-aging operator describes this coordinated pattern of change. At a given age, it represents how the fitted covariance structure changes with an incremental increase in age. Individual entries describe changes in covariance between particular regions, but the operator retains those changes as one multivariate object. An edge-by-edge analysis treats those changes as separate effects, whereas the operator-valued analysis characterizes their joint direction, magnitude, and effective dimensionality.

Panel B of Figure~\ref{fig:dlbs_brain} combines posterior summaries from the fitted model at representative ages with a schematic network-level illustration of within-person covariance aging. Node colors correspond to the same seven functional networks shown in Panel A, and the network anchors are positioned schematically to reflect their approximate anatomical locations. The displayed connections are schematic rather than parcel-level posterior edge estimates; within the illustration, blue and red connections denote negative and positive covariance-scale change, respectively. The accompanying posterior summaries quantify how the organization of longitudinal covariance change evolves across adulthood.

\begin{figure}[!htbp]
\centering
\includegraphics[width=\textwidth]{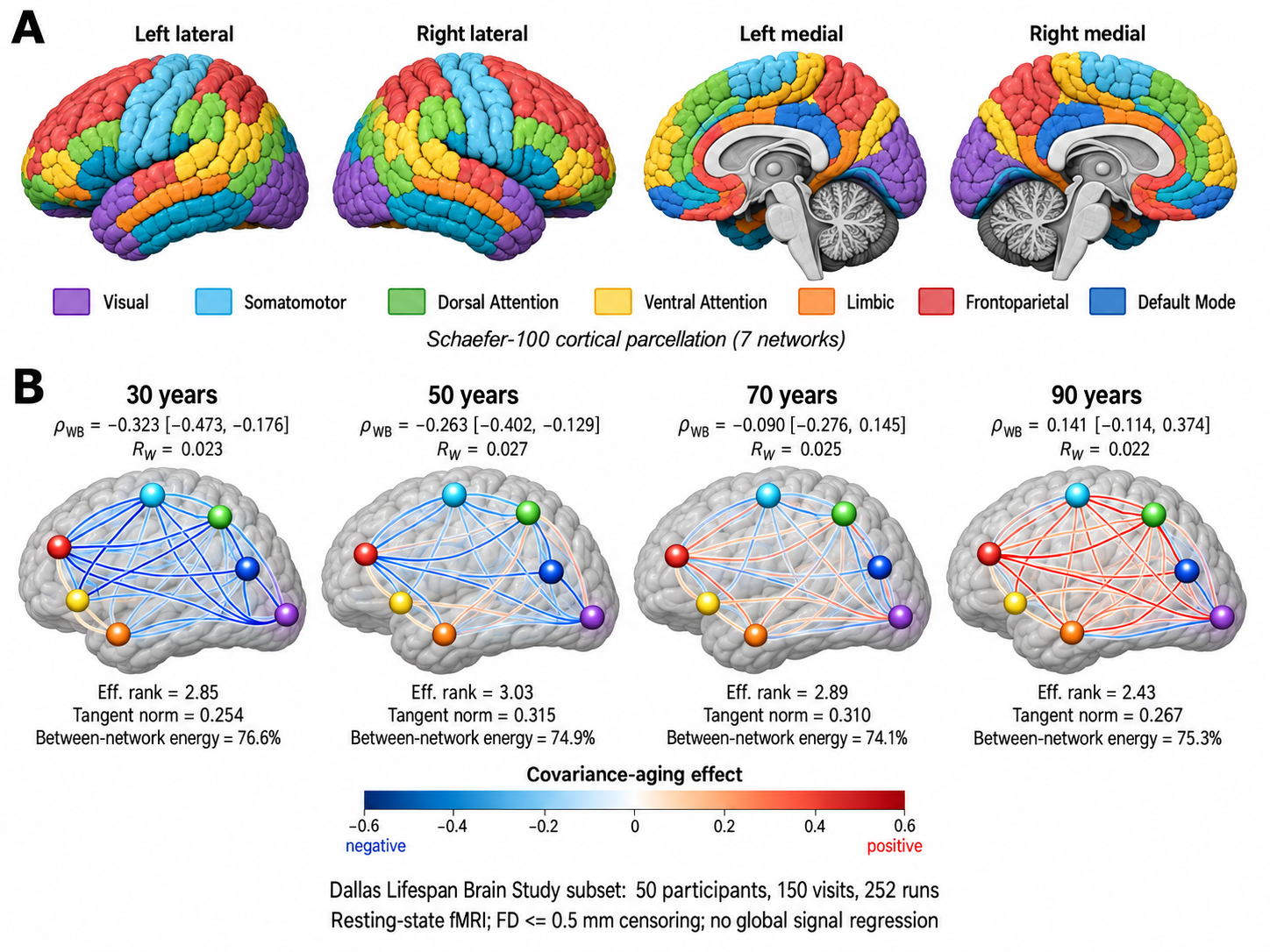}
\caption{\textbf{Cortical domain and longitudinal covariance aging in the DLBS.}
(A) Schaefer-100 cortical parcellation grouped into seven functional networks.
(B) Posterior summaries at ages 30, 50, 70, and 90 years shown alongside a schematic network-level representation of within-person covariance aging. Node colors match the networks in Panel A; node positions and displayed connections are schematic rather than parcel-level posterior estimates. Posterior summaries quantify within- versus between-person alignment, reorganization rate, effective rank, and between-network energy.}
\label{fig:dlbs_brain}
\end{figure}

The figure also illustrates why the distinction between cross-sectional and longitudinal aging is made at the operator level. Aging may weaken covariance within one network while strengthening coupling between others, or reorganize dominant spatial modes without producing a uniform change in connectivity. These patterns can correspond to distinct directions of whole-system covariance change even when simpler connectivity summaries are similar. The method developed below therefore targets the direction, magnitude, and effective dimensionality of this coordinated reorganization.

The scientific target can now be stated precisely. Let $a_{ij}$ denote age for participant $i$ at visit $j$, and let
\[
\bar a_i = \frac{1}{n_i}\sum_{j=1}^{n_i} a_{ij}
\]
denote that participant's mean observed age. We decompose age as
\[
a_{ij} = \bar a_i + w_{ij},
\qquad
w_{ij}=a_{ij}-\bar a_i.
\]
The participant mean $\bar a_i$ captures between-person location in the age distribution, while $w_{ij}$ captures within-person change around that participant's observed age center. The use of $\bar a_i$ is a statistical centering device rather than a biological reference age. A negative value of $w_{ij}$ simply indicates a visit occurring before that participant's observed midpoint, and $w_{ij}$ still increases one-for-one with chronological age. Thus, the direction of the within-person age effect retains its usual chronological interpretation.

Person-mean centering also provides a direct decomposition of observed age into between-person and within-person components. Using baseline age $a_{i1}$ instead would anchor the decomposition to study entry, making the between-person component depend on the particular visit designated as baseline and forcing all subsequent within-person values to be positive. Centering at $\bar a_i$ instead uses all observed visits to define each participant's age location and places the repeated measurements around zero, providing a more symmetric separation of longitudinal change from between-person age differences.

The central question is whether the covariance change associated with increasing $w_{ij}$ follows the same direction as the covariance difference associated with increasing $\bar a_i$. We also ask whether the within-person direction changes across adulthood and whether the complexity of that direction itself varies with age. These questions are operator-valued because they concern coordinated changes in the covariance system as a whole. That perspective is biologically relevant because loss of network segregation, redistribution of covariance strength, and reorientation of dominant functional modes are inherently multivariate phenomena. It is statistically relevant because treating thousands of pairwise connections as unrelated responses discards both the geometry and the shared structure of the covariance object.

\section{Longitudinal covariance model}
\label{sec:model}

We seek to describe how the covariance of brain activity changes as an individual grows older. This requires two distinctions. First, an estimated covariance from a BOLD run is an uncertain measurement, not the underlying covariance state. Second, differences between people of different ages do not necessarily describe changes within the same person. Our approach addresses both questions in one model. We first calibrate the information in each run, then model longitudinal and between-person age effects separately. A central contribution is that the directions of longitudinal covariance change are learned from the repeated visits rather than selected from the dominant patterns of overall covariance variation.

\subsection{Accounting for information in each run}

Let $i$ index participants, $j$ visits, and $r$ runs. The $T_{ijr}$ retained BOLD measurements from $p$ brain regions produce a sample covariance matrix $\bm S_{ijr}$. We write $\bm\Sigma_{ij}$ for the underlying covariance at that visit. Because successive BOLD measurements are correlated, $T_{ijr}$ overstates the number of independent observations. We use the temporal correlation matrix $\bm R_{ijr}$, evaluated at the original retained frame locations, to calculate the effective covariance information $\nu_{ijr}=T_{ijr}^{2}/\tr(\bm R_{ijr}^{2})$. This accounts for temporal dependence and the gaps created by censoring.

Repeated runs at a visit provide a second source of information about measurement uncertainty. We estimate their additional covariance dispersion using a common factor $\phi$ and reduce the run information to $\nu_{ijr}^{\ast}=\nu_{ijr}/\phi$. In DLBS, $\widehat\phi=3.5805$, giving an average of 17.45 effective covariance observations per run after both adjustments. We then weight each run's covariance contribution according to this calibrated information. The resulting covariance quasi-likelihood has the familiar parameter-dependent form of a Wishart likelihood, but it does not require an exact Wishart sampling distribution. In particular, the effective information can be smaller than $p$. 

This calibration is part of the Bayesian analysis, not a correction applied after fitting. Runs supported by less independent information contribute less precision to the posterior, while repeated-run variability prevents the model from treating noisy covariance estimates as exact observations.

\subsection{Modeling change within and between people}

We model a run covariance through a positive-definite reference matrix and a symmetric covariance-change operator. The full model is
\begin{equation}
\begin{aligned}
\bm\Sigma_{ijr}&=\bm\Sigma_{\star}^{1/2}\bm D_{ij}^{1/2}
\exp\{\bm A_{ijr}\}\bm D_{ij}^{1/2}\bm\Sigma_{\star}^{1/2},\\
\bm A_{ijr}&=w_{ij}\bm B_W(\bar a_i)+(\bar a_i-a_0)\bm B_B\\
&\quad+m_{ijr}\bm B_M+s_i\bm B_S+\bm U_i+\bm V_{ij}.
\end{aligned}
\label{eq:tangentdecomp}
\end{equation}
Here $\bm\Sigma_{\star}$ is a fixed, positive-definite reference covariance and $\bm D_{ij}$ is a positive diagonal matrix that allows marginal covariance amplitudes to vary by participant and visit. The symmetric matrix $\bm A_{ijr}$ describes coordinated changes in covariance structure. The matrix exponential guarantees that the resulting run covariance $\bm\Sigma_{ijr}$ is positive definite.

The age terms separate longitudinal change from population differences. We write $a_{ij}$ for the participant's age at visit $j$, $\bar a_i$ for their mean observed age, and $w_{ij}=a_{ij}-\bar a_i$ for their within-person age displacement. The reference age $a_0$ centers the between-person effect. The age-varying symmetric operator $\bm B_W(a)$ describes how covariance changes within a person at age $a$. The symmetric operator $\bm B_B$ describes how covariance differs across people with different mean ages. Estimating these operators separately is essential because the cross-sectional gradient may differ substantially from the longitudinal trajectory.

The remaining terms account for variation that should not be attributed to aging. The covariate $m_{ijr}$ is standardized run-level mean framewise displacement, with coefficient operator $\bm B_M$. The centered sex covariate $s_i$ has coefficient operator $\bm B_S$. The zero-centered symmetric random operators $\bm U_i$ and $\bm V_{ij}$ represent persistent participant differences and visit-specific variation. Their regularization is separate from that of the age effects. In this way, the model allows rich covariance heterogeneity without confusing it with systematic longitudinal change.

The within-person operator $\bm B_W(a)$ varies smoothly across adulthood. We express it using smooth age-basis weights and symmetric coefficient matrices in a learned span, so its strength and dominant directions can evolve with age. The between-person operator is learned in a separate span. The spans may overlap, since their effects are distinguished by the age variables and the repeated-measures structure rather than by requiring spatial orthogonality.

\subsection{Learning the structure of aging}

Although a covariance matrix has many entries, reproducible aging may involve only a small number of coordinated directions. We allow a flexible set of candidate aging directions and place cumulative ordered-shrinkage priors on their coefficients. Later directions receive progressively stronger shrinkage, but none is discarded in advance. The posterior therefore learns the effective complexity of the aging signal while retaining broader participant and visit variation. This differs from selecting principal components based on pooled covariance, which can emphasize between-person differences rather than longitudinal aging.

The low-dimensional representation also makes fitting practical. Matrix exponentials are evaluated exactly within the combined span of the systematic and nuisance components. Technical details on run information and the low-dimensional exponential identity are provided in the Supplementary Materials. The important point is that this reduction changes neither the covariance model nor the meaning of the aging operators. Section~\ref{sec:method} now turns those operators into quantities that describe the magnitude and organization of brain aging.

\section{Quantifying longitudinal covariance aging}
\label{sec:method}

The model estimates covariance change through symmetric operators, but those operators are expressed relative to a reference covariance. To interpret aging biologically, we translate the fitted trajectory back to the covariance scale. This provides a direct comparison between change within a person and age-related differences across people, along with a measure of how rapidly the covariance state is reorganized.

At each age $a$, let $\bm D_W(a)$ denote the derivative of the fitted covariance with respect to within-person age, and let $\bm D_B(a)$ denote the corresponding derivative for between-person age. Both are symmetric matrices on the original covariance scale, obtained using the exact derivative of the matrix exponential. We compare their directions through
\begin{equation}
\rho_{WB}(a)=\frac{\langle\bm D_W(a),\bm D_B(a)\rangle_F}
{\norm{\bm D_W(a)}_F\norm{\bm D_B(a)}_F}.
\label{eq:alignment}
\end{equation}
Here $\langle\cdot,\cdot\rangle_F$ is the Frobenius inner product and $\norm{\cdot}_F$ is its associated matrix norm. Values near one indicate similar covariance-change patterns, values near zero indicate weak alignment, and negative values indicate disagreement in their overall directions. This comparison concerns the complete covariance pattern, not a uniform sign change in every connection.

We also measure the overall rate of longitudinal reorganization by $R_W(a)=\norm{\bm D_W(a)}_F$. Its variation with age reveals when the covariance state changes more quickly. To distinguish changes in the strength of existing covariance patterns from changes in their spatial organization, we examine the fitted spectral structure. We summarize the number of directions carrying meaningful longitudinal change using the effective participation rank $r_{\mathrm{eff}}(a)$. A small value indicates concentrated change, while a larger value indicates that change is spread across more directions. These measures are evaluated for posterior draws, so their reported uncertainty includes uncertainty about the underlying covariance trajectory.

\subsection{Locating changes across functional networks}

Whole-brain measures tell us whether the two age effects agree, but not where their differences arise. We therefore use the seven-network Schaefer partition \citep{schaefer2018} to group covariance changes by functional system. For each network pair, we estimate the share of longitudinal covariance-change energy and assess the alignment of longitudinal and between-person changes within that block. This distinguishes changes within networks from changes between networks and identifies the systems contributing most to global disagreement.

The network summaries are calculated from the same posterior covariance derivatives used for the whole-brain analysis. They are therefore descriptive decompositions of one fitted model, rather than separate network models or thousands of connection-level tests. The derivative and network summaries are defined using the fitted covariance trajectory.

\section{Bayesian inference and theoretical properties}
\label{sec:inference}

Bayesian inference is particularly useful here because the covariance state at each visit is uncertain and the direction of longitudinal aging is learned rather than specified in advance. We combine the calibrated run-level quasi-likelihood with priors on the age effects and nuisance components. Smooth age-basis coefficients allow the longitudinal operator to evolve gradually, while ordered shrinkage favors a small number of reproducible aging directions. Separate centered priors on participant and visit effects accommodate variation that is not part of the systematic aging trajectory.

The result is a generalized Bayesian posterior over covariance trajectories and their age-dependent operators. Posterior draws are propagated through the covariance-scale calculations in Section~\ref{sec:method} to obtain estimates and credible intervals for alignment, rate, effective rank, and network summaries. Because the updating criterion is a calibrated quasi-likelihood rather than an exact sampling likelihood, these are generalized Bayesian credible intervals. Their frequentist coverage is assessed separately in the simulation study. This distinction is important: calibration of run-level covariance information is intended to improve uncertainty quantification, but it does not by itself guarantee nominal coverage for nonlinear summaries.

Our main structural result is that the model always produces valid covariance matrices. This follows directly from the matrix exponential representation.

\begin{theorem}[Positive-definite covariance trajectory]
Let $\bm\Sigma_{\star}$ be positive definite, $\bm D$ a positive diagonal matrix, and $\bm A$ symmetric. Then
\[
\bm\Sigma=\bm\Sigma_{\star}^{1/2}\bm D^{1/2}
\exp(\bm A)\bm D^{1/2}\bm\Sigma_{\star}^{1/2}
\]
is positive definite. If the model components vary smoothly with age, the resulting covariance trajectory is smooth.
\end{theorem}

The covariance operators and their scientific summaries are also unchanged by arbitrary changes of basis within a learned span. This invariance concerns the represented covariance trajectory, not a claim that every coordinate-specific prior is rotation invariant. Together, these properties provide a valid covariance trajectory whose longitudinal behavior can be estimated and interpreted with uncertainty.

A simulation study was conducted to evaluate the performance and validity of the proposed method. The results demonstrated accurate recovery of the true within-person covariance-aging direction and showed that the method successfully distinguishes longitudinal aging from between-person age differences. These findings provide strong evidence of the method's ability to recover the underlying covariance-aging structure and support its use for studying longitudinal brain reorganization. The proof, technical justifications, and simulation results are provided in the Supplementary Materials.

\section{DLBS analysis}
\label{sec:dlbs}

The DLBS application returns to the biological question that motivated the model: whether the covariance pattern associated with age across people is the same pattern expressed as people themselves age. The analysis included 50 participants observed at three visits, for 150 visits and 252 resting-state runs spanning ages 22--96 years. All validated run-level time-series matrices contained the same 100 Schaefer cortical parcels. Runs retained 89--180 volumes after censoring and nuisance processing. The analysis used percent signal change rather than z-standardization so that covariance amplitude remained part of the target. The run-level sample covariance matrices were themselves structured, but their dominant pooled directions were not used to define the longitudinal aging span. An overview of the age-specific posterior covariance-aging patterns is shown in Figure~\ref{fig:dlbs_brain}B, with the individual inferential summaries examined below.

Repeated runs were important for separating uncertainty in an observed scan from genuine visit-to-visit covariance change. Temporal autocorrelation reduced the nominal covariance information substantially, and the additional repeated-run calibration gave $\widehat\phi=3.5805$. Under this correction, the mean effective covariance information was 17.45 per run. Earlier posterior covariance checks under the same statistical hierarchy gave a pooled standardized residual standard deviation close to one and approximately 94\% coverage by nominal 95\% covariance-moment intervals, supporting the scale of the calibrated quasi-likelihood before the biological functionals were interpreted.

\subsection{Within-person covariance aging is not the cross-sectional age gradient}

The clearest biological result concerns the distinction between longitudinal aging and population age differences. Figure~\ref{fig:alignment} and Table~\ref{tab:aging-measures} summarize the posterior alignment $\rho_{WB}(a)$ between the within-person covariance derivative and the corresponding between-person age derivative. The posterior mean was $-0.323$ at age 30, $-0.304$ at age 40, $-0.263$ at age 50, and $-0.194$ at age 60. The 95\% credible intervals remained entirely below zero through age 60. At age 70 the posterior mean was $-0.090$ and the interval included zero; by ages 80 and 90 the posterior means were close to zero or mildly positive.

This result indicates that, through early and middle adulthood, the covariance direction followed by an individual as that individual ages is not simply the direction obtained by comparing older and younger adults. The negative alignment is geometric rather than a claim that every connection changes with the opposite sign. It means that the dominant multivariate covariance components of longitudinal aging are partly opposed to those of the cross-sectional age gradient. The discrepancy weakens in later life. In the context of the DLBS literature, where reduced functional-system segregation was established primarily from cross-sectional age contrasts, this finding cautions against treating a population age signature as the trajectory of an individual brain.
{}

\begin{figure}[t]
\centering
\includegraphics[width=.72\textwidth]{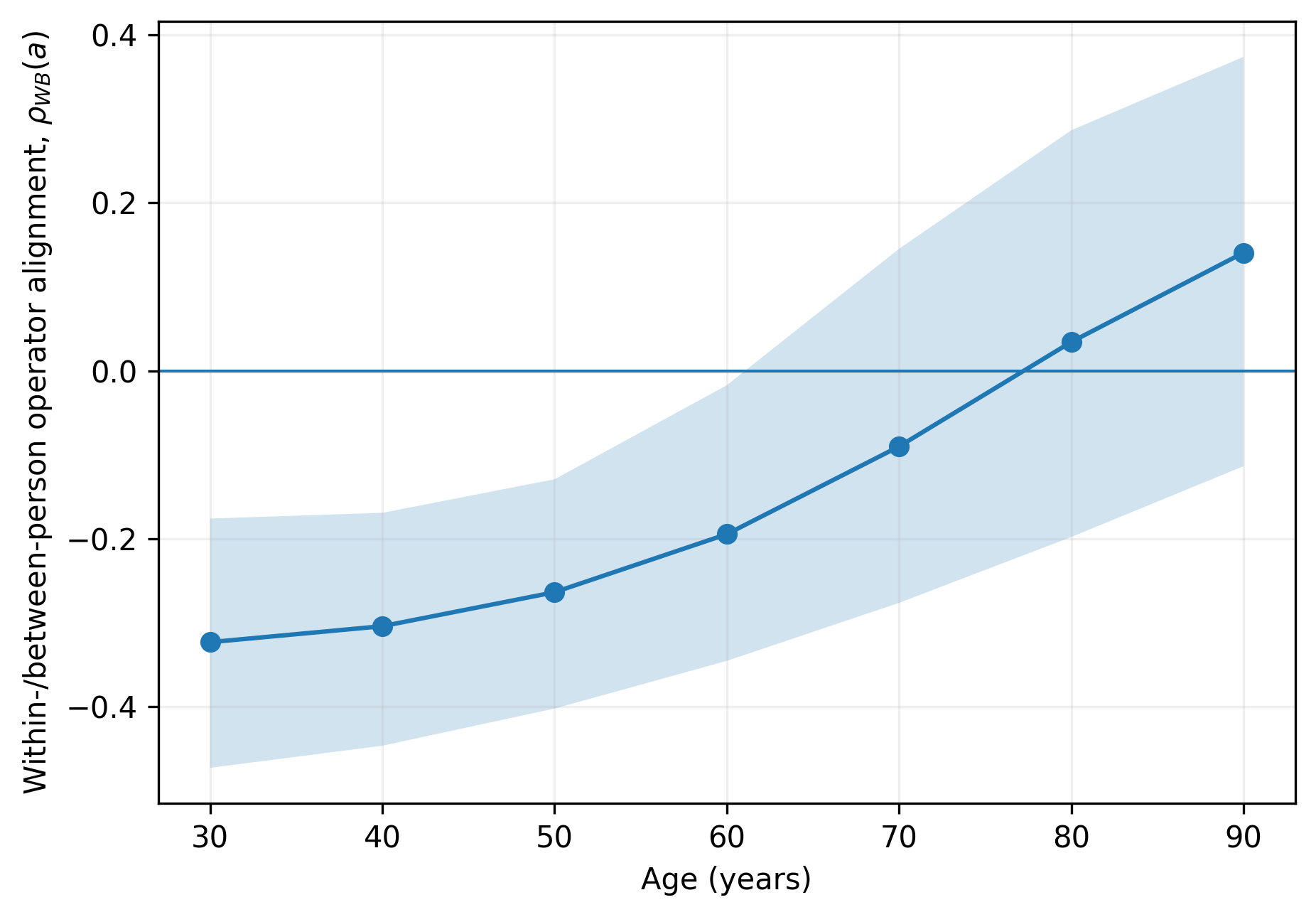}
\caption{Posterior alignment between the within-person covariance-aging derivative and the between-person age derivative. The line is the posterior mean of $\rho_{WB}(a)$ and the band is the pointwise 95\% credible interval. Negative values through early and middle adulthood indicate that the dominant covariance direction of longitudinal aging differs from the cross-sectional age gradient.}
\label{fig:alignment}
\end{figure}

\begin{table}[t]
\centering
\caption{Posterior summaries of the principal measures of longitudinal covariance aging. Values are posterior means with 95\% credible intervals in parentheses.}
\label{tab:aging-measures}
\small
\begin{tabular}{rccc}
\toprule
Age & $R_W(a)$ & $\rho_{WB}(a)$ & $r_{\mathrm{eff}}(a)$ \\
\midrule
30 & 0.023 (0.017, 0.031) & -0.323 (-0.473, -0.176) & 2.85 (2.35, 3.40) \\
40 & 0.025 (0.019, 0.033) & -0.304 (-0.446, -0.169) & 3.06 (2.58, 3.54) \\
50 & 0.027 (0.021, 0.033) & -0.263 (-0.402, -0.129) & 3.03 (2.60, 3.51) \\
60 & 0.027 (0.022, 0.034) & -0.194 (-0.345, -0.017) & 2.97 (2.50, 3.45) \\
70 & 0.025 (0.020, 0.032) & -0.090 (-0.276, 0.145) & 2.89 (2.45, 3.35) \\
80 & 0.023 (0.018, 0.030) & 0.034 (-0.198, 0.286) & 2.69 (2.23, 3.27) \\
90 & 0.022 (0.016, 0.029) & 0.141 (-0.114, 0.374) & 2.43 (1.89, 3.07) \\
\bottomrule
\end{tabular}
\end{table}

\subsection{The rate of covariance reorganization is greatest near midlife}

The global reorganization rate $R_W(a)$ showed a smooth, nonmonotone trajectory (Figure~\ref{fig:rate}). Its posterior mean increased from 0.023 at age 30 to approximately 0.027 at ages 50--60 and then declined to 0.022 by age 90. The result does not support a picture in which whole-brain covariance simply changes progressively faster with advancing age. Instead, the estimated rate is greatest around midlife and attenuates thereafter. Because $R_W(a)$ is a norm, its positive credible interval is not interpreted as a formal test against an atom at zero; the scientific information is in the age-varying magnitude and its uncertainty.

\begin{figure}[t]
\centering
\includegraphics[width=.72\textwidth]{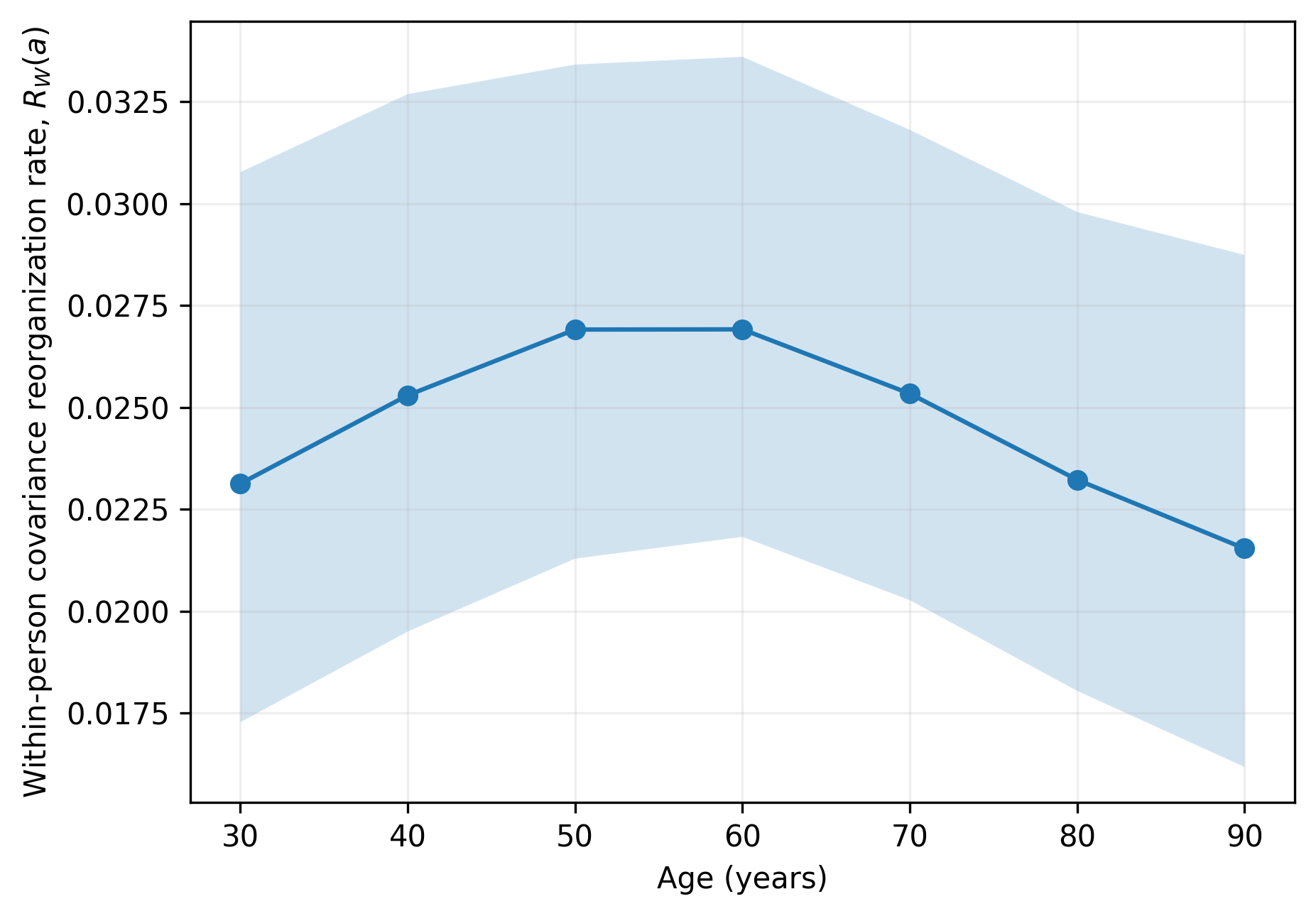}
\caption{Posterior rate of within-person covariance reorganization, $R_W(a)=\|D_W(a)\|_F$. The estimated rate rises modestly into midlife and then declines toward later adulthood. Bands are pointwise 95\% credible intervals under the primary low-rank variational approximation.}
\label{fig:rate}
\end{figure}

\subsection{Longitudinal aging is predominantly between systems and the discrepancy is visual-centered}

The operator-level disagreement has a specific biological organization rather than being diffusely distributed across the connectome. Across the entire age range, 74.1--76.6\% of the off-diagonal energy in $\bm D_W(a)$ was carried by connections linking different functional systems (Table~\ref{tab:network-localization}). Thus, the longitudinal aging signal is primarily a reorganization of coordination \emph{among} large-scale systems, not simply weakening or strengthening within isolated networks. This proportion was strikingly stable across adulthood even though the direction of longitudinal aging relative to the cross-sectional gradient changed substantially.

The early- and mid-adult longitudinal--cross-sectional discrepancy was concentrated around the visual system (Figure~\ref{fig:brain-network-localization}). At age 30 the global off-diagonal cosine between $\bm D_W(a)$ and $\bm D_B(a)$ was $-0.317$ (95\% interval $[-0.470,-0.163]$). The strongest negative network-pair contributions involved Visual--Visual, Visual--Default, Visual--Control, Visual--Limbic, and Visual--Dorsal Attention relationships. The corresponding block cosines for Visual--Visual, Visual--Default, and Visual--Limbic were $-0.641$, $-0.503$, and $-0.732$. The same visual-centered ordering remained apparent at ages 40--60, while the global discrepancy gradually weakened.
{}

A particularly informative result appears in later life. By ages 80--90 the global off-diagonal cosine was near zero or mildly positive, indicating that the whole covariance-age directions were no longer clearly discordant. That global summary, however, masks a persistent network-specific difference. Visual--Limbic alignment remained strongly negative, with posterior mean block cosines of $-0.671$ at age 80 and $-0.611$ at age 90; the posterior probability of a negative Visual--Limbic cosine at age 90 was 0.993. Limbic--Default alignment also remained negative in later adulthood. Figure~\ref{fig:visual-persistence} shows that the visual--limbic discrepancy persists even as the global longitudinal--cross-sectional disagreement attenuates.

This pattern adds biological content that cannot be obtained from the whole-operator norm alone. The same low-dimensional longitudinal aging process is expressed mainly through between-system covariance, and the relationship between individual aging and population age differences changes over the lifespan in a network-specific manner. The result does not imply that cross-sectional reductions in system segregation are incorrect. It shows that the covariance trajectory followed within individuals is a different object, with a visual-centered discrepancy in earlier adulthood and selective persistence of visual--limbic discordance after the global discrepancy has largely resolved.

\begin{figure}[t]
\centering
\includegraphics[width=.96\textwidth]{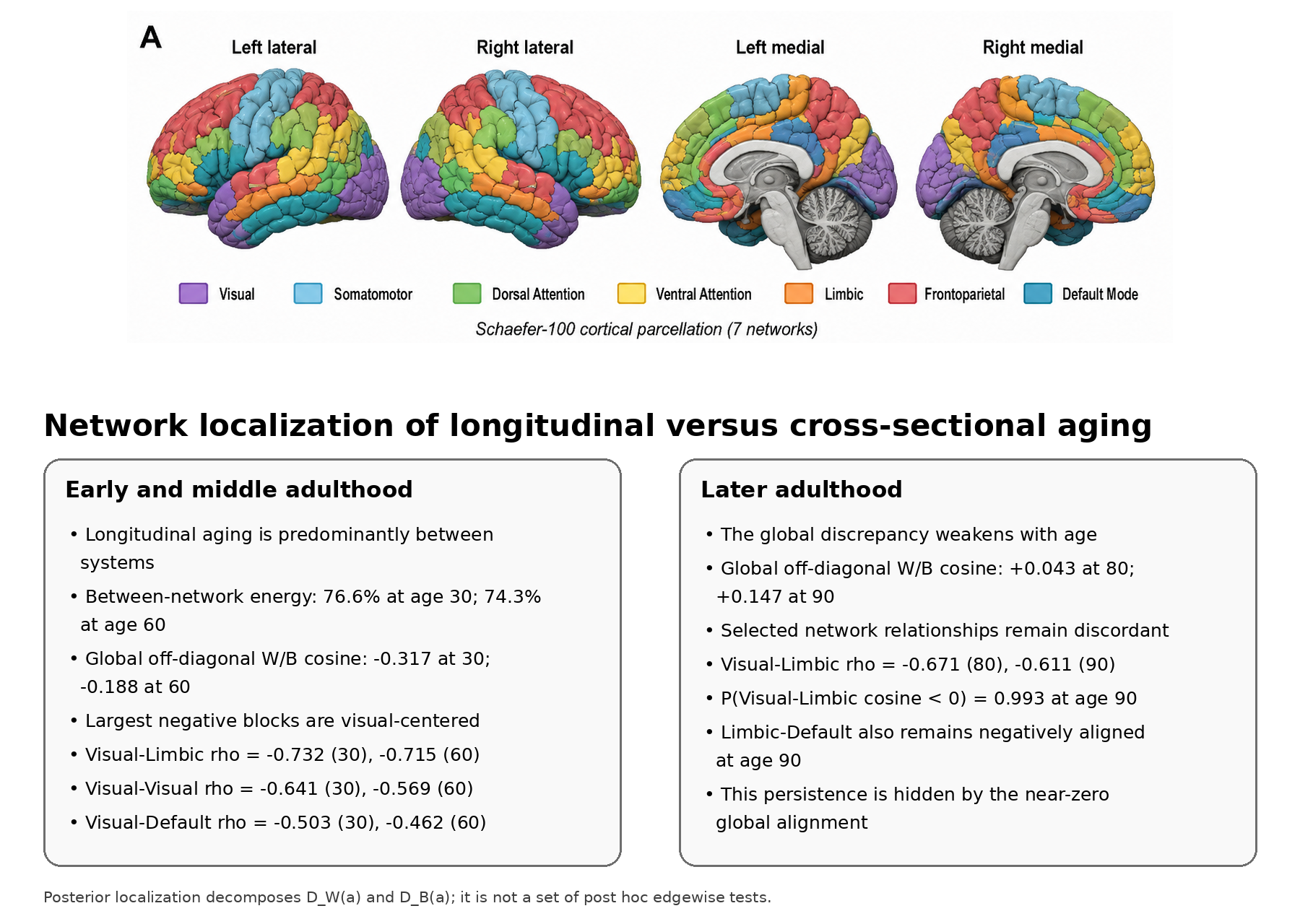}
\caption{Brain-network localization of the longitudinal--cross-sectional covariance-aging discrepancy. The Schaefer-100 seven-network cortical organization provides the anatomical reference. Posterior decomposition of the off-diagonal covariance-scale derivatives shows that roughly three quarters of longitudinal aging energy is between systems. The strongest early- and mid-adult disagreement is visual-centered; in later life the global discrepancy weakens, while Visual--Limbic discordance remains strongly negative. The network results are posterior decompositions of the primary operators, not selected edgewise tests.}
\label{fig:brain-network-localization}
\end{figure}

\begin{figure}[t]
\centering
\includegraphics[width=.94\textwidth]{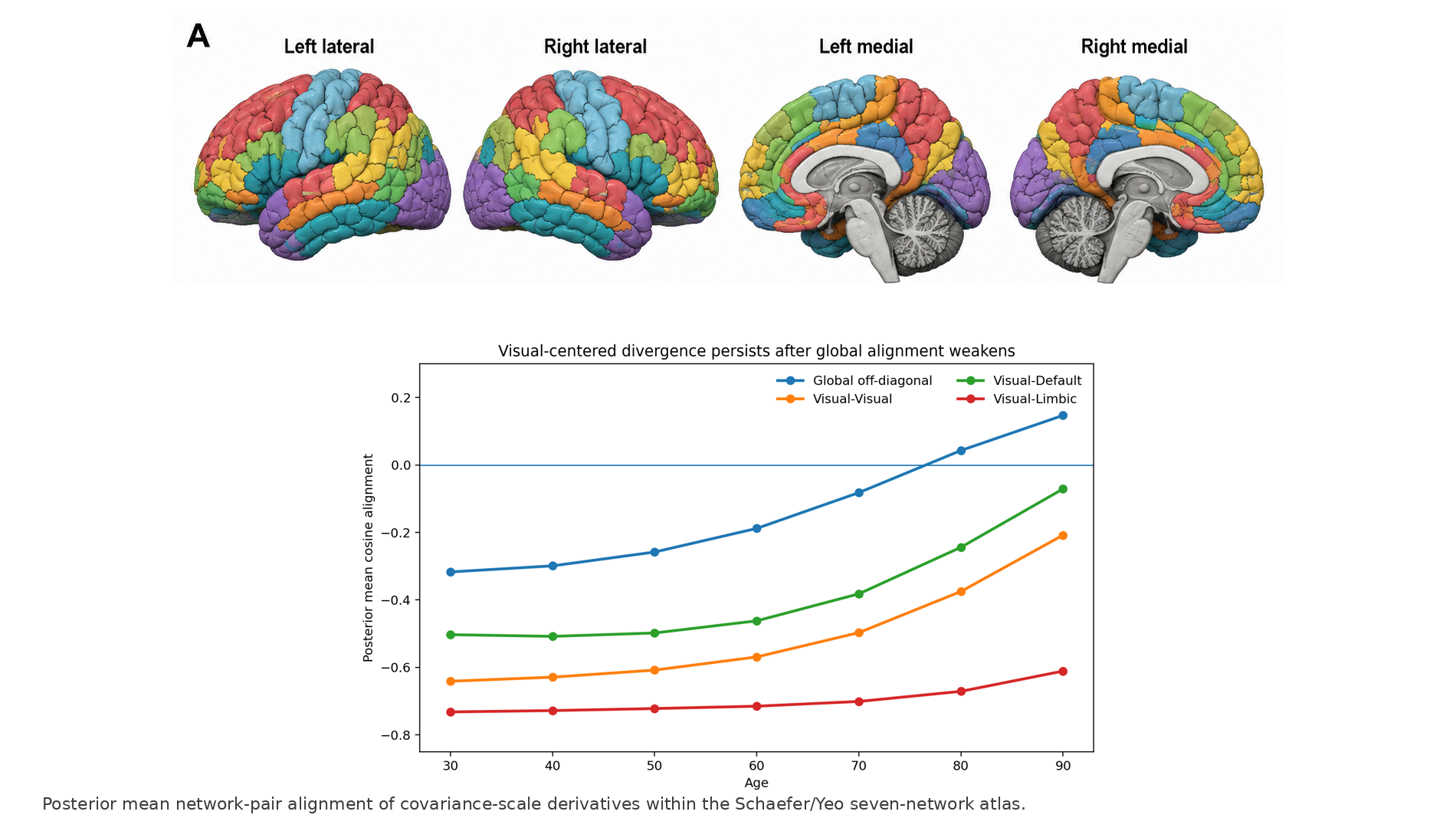}
\caption{Age evolution of the network relationships that most clearly distinguish longitudinal from cross-sectional covariance aging. The global off-diagonal alignment moves from negative in early adulthood toward zero and mildly positive values in later life. Visual--Visual and Visual--Default disagreement attenuate in parallel, whereas Visual--Limbic alignment remains strongly negative across the observed lifespan. Lines are posterior means from the network-resolved decomposition.}
\label{fig:visual-persistence}
\end{figure}

\begin{table}[t]
\centering
\caption{Network-resolved localization of the longitudinal--cross-sectional covariance-aging discrepancy. The global cosine and block cosines are posterior means computed from off-diagonal covariance-scale derivatives. Between-network energy is the percentage of within-person derivative energy carried by connections linking different Schaefer/Yeo systems.}
\label{tab:network-localization}
\small
\begin{tabular}{rrrrrr}
\toprule
Age & Global cosine & Between-network energy & Visual--Visual & Visual--Default & Visual--Limbic\\
\midrule
30 & -0.317 & 76.6\% & -0.641 & -0.503 & -0.732\\
40 & -0.299 & 75.7\% & -0.629 & -0.508 & -0.728\\
50 & -0.258 & 74.9\% & -0.608 & -0.498 & -0.722\\
60 & -0.188 & 74.3\% & -0.569 & -0.462 & -0.715\\
70 & -0.082 & 74.1\% & -0.497 & -0.382 & -0.701\\
80 &  0.043 & 74.5\% & -0.375 & -0.244 & -0.671\\
90 &  0.147 & 75.3\% & -0.208 & -0.071 & -0.611\\
\bottomrule
\end{tabular}
\end{table}

\subsection{Systematic aging is low dimensional but its direction evolves smoothly}

The systematic longitudinal aging signal required only a small number of coordinated covariance directions. The posterior effective rank was approximately 2.4--3.1 across the age range (Figure~\ref{fig:rank}), even though covariance variation across individual scans was considerably richer. Ordered shrinkage provided an independent check on this conclusion: the first four candidate within-person components remained strongly active, the fifth was substantially reduced, and the posterior shrinkage scale declined to 0.002 by component 20 (Figure~\ref{fig:shrinkage}). The low-dimensional conclusion therefore concerns the reproducible population aging signal, not the dimensionality of individual covariance heterogeneity.

\begin{figure}[t]
\centering
\includegraphics[width=.72\textwidth]{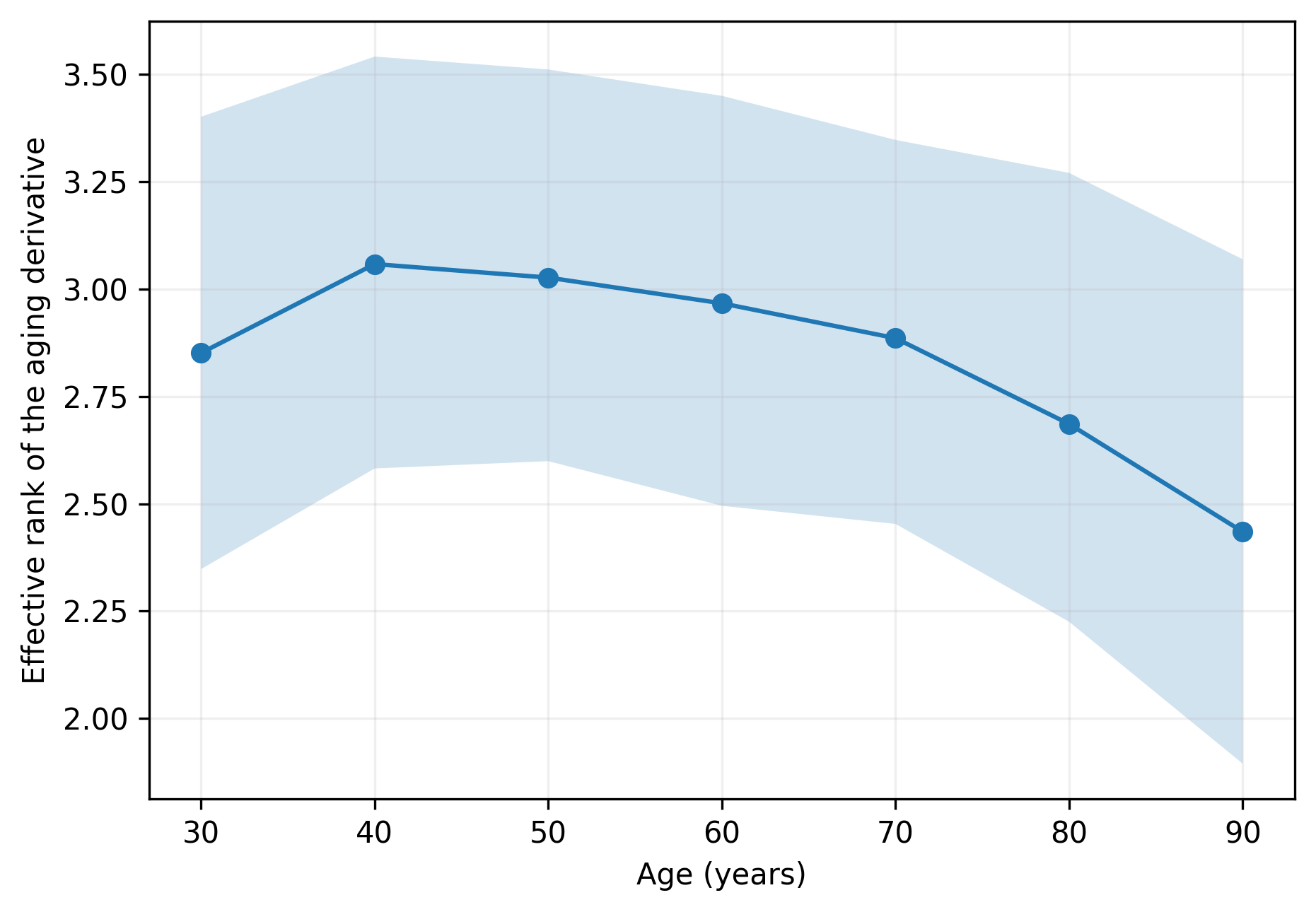}
\caption{Posterior effective rank of the systematic within-person covariance-aging derivative. The aging signal remains concentrated in roughly three effective directions across adulthood.}
\label{fig:rank}
\end{figure}

\begin{figure}[t]
\centering
\includegraphics[width=.72\textwidth]{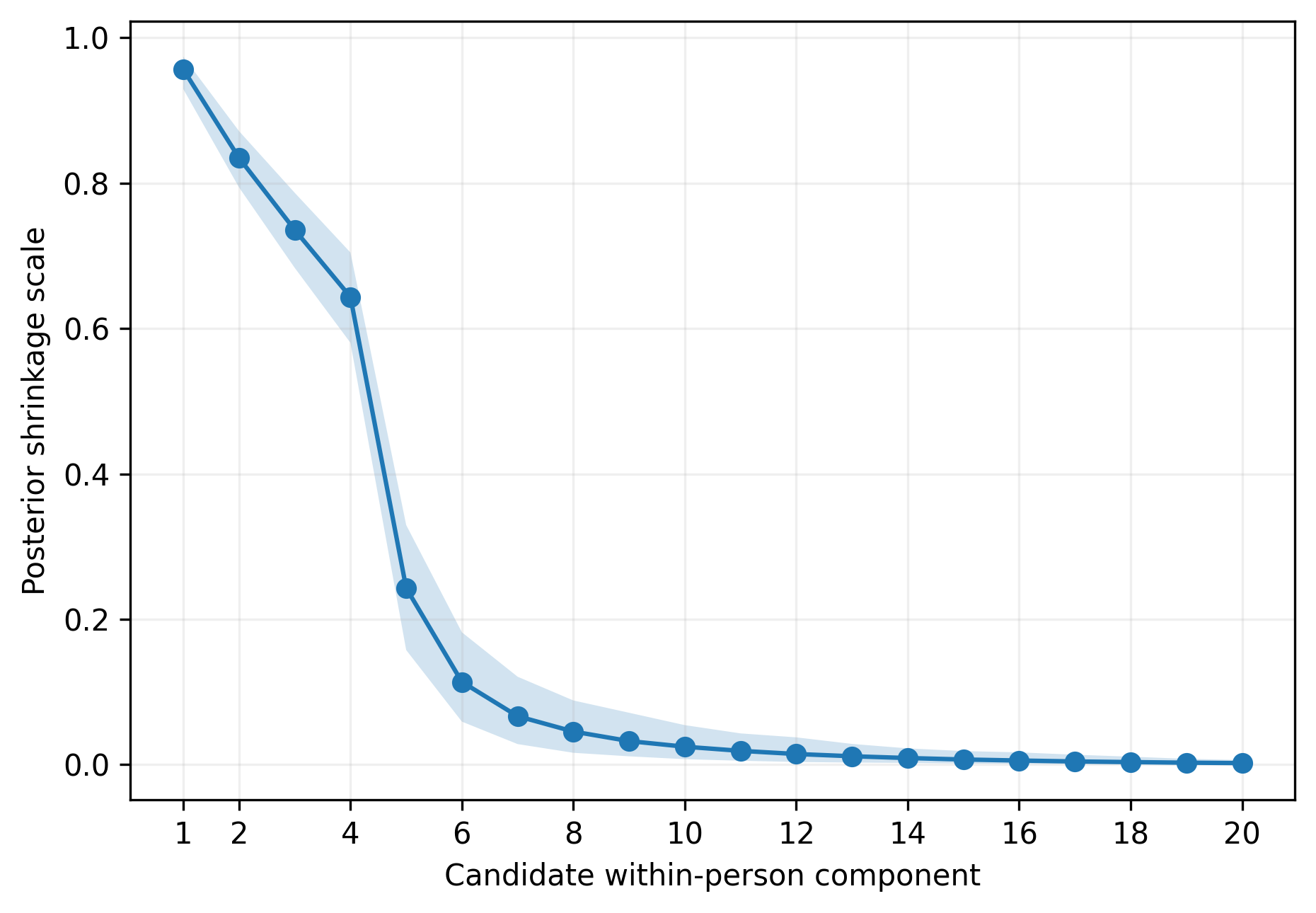}
\caption{Posterior ordered-shrinkage profile for the 20 candidate within-person aging directions. The first four directions dominate, while later directions are progressively suppressed and the truncation boundary is essentially inactive.}
\label{fig:shrinkage}
\end{figure}

Low dimensionality does not imply that one fixed aging operator describes the entire adult lifespan. Nearby ages had extremely similar longitudinal operators: posterior mean cosine similarity was 0.989 for ages 30 and 40, 0.994 for ages 50 and 60, and 0.991 for ages 70 and 80. Similarity declined gradually with wider age separation, reaching 0.847 for ages 30 and 90 (Figure~\ref{fig:crossage}; Table~\ref{tab:crossage}). The pattern is therefore one of local stability with gradual reorientation rather than abrupt rotation. The form of covariance aging evolves over adulthood while remaining smooth from one decade to the next.

\begin{figure}[t]
\centering
\includegraphics[width=.72\textwidth]{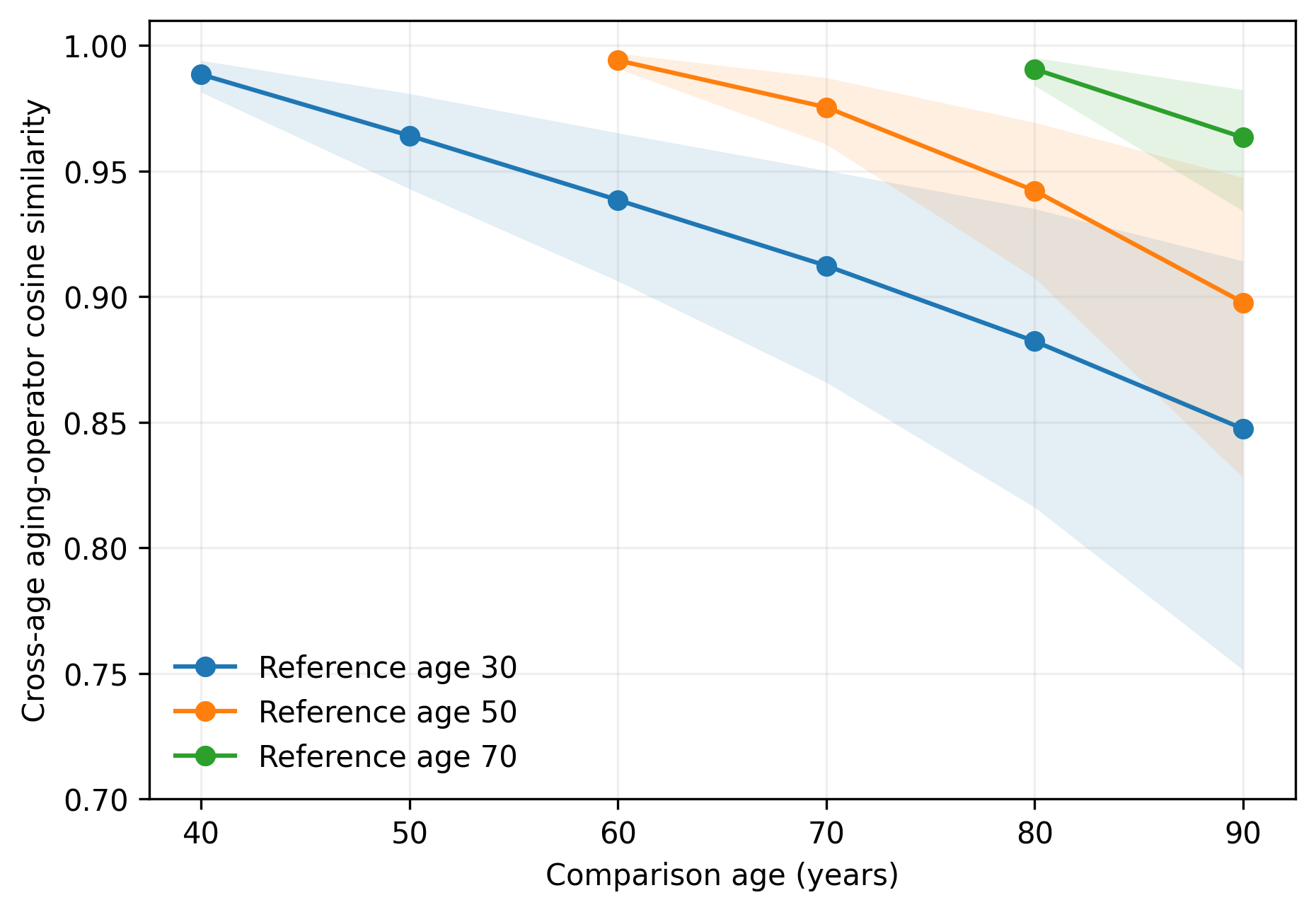}
\caption{Selected posterior cosine similarities between within-person aging operators evaluated at different ages. Operators are highly similar locally but become progressively less aligned over longer age spans, supporting smooth lifespan reorientation.}
\label{fig:crossage}
\end{figure}

\begin{table}[t]
\centering
\caption{Selected posterior cross-age similarities of the within-person aging operator. Values are posterior mean cosine similarities with 95\% credible intervals in parentheses.}
\label{tab:crossage}
\small
\begin{tabular}{rrc}
\toprule
Reference age & Comparison age & Cosine similarity \\
\midrule
30 & 40 & 0.989 (0.981, 0.994) \\
30 & 50 & 0.964 (0.943, 0.981) \\
30 & 60 & 0.939 (0.906, 0.965) \\
30 & 70 & 0.912 (0.866, 0.950) \\
30 & 80 & 0.882 (0.816, 0.935) \\
30 & 90 & 0.847 (0.751, 0.914) \\
50 & 60 & 0.994 (0.991, 0.997) \\
50 & 70 & 0.975 (0.961, 0.987) \\
50 & 80 & 0.942 (0.907, 0.969) \\
50 & 90 & 0.898 (0.828, 0.947) \\
70 & 80 & 0.991 (0.984, 0.995) \\
70 & 90 & 0.963 (0.934, 0.982) \\
\bottomrule
\end{tabular}
\end{table}

\subsection{Interpretation and computational sensitivity}

Taken together, the results identify three features of longitudinal functional-covariance aging. The within-person trajectory differs from the cross-sectional age gradient through early and middle adulthood; the overall rate of reorganization is greatest around midlife rather than at the oldest ages; and the systematic aging signal is low dimensional while its orientation changes gradually across adulthood. These findings are complementary. A low-dimensional process can still evolve in direction, and a cross-sectional covariance signature can remain biologically informative without being an unbiased surrogate for within-person aging.

The network-resolved analysis strengthens the connection between the methodological and biological contributions. The primary finding is defined at whole-operator scale, where the model separates within-person from between-person covariance aging. The secondary decomposition then shows why that operator distinction matters biologically: approximately three quarters of longitudinal aging energy is carried by between-system covariance, the early- and mid-adult discrepancy is visual-centered, and Visual--Limbic discordance persists after the global alignment has approached zero. None of these conclusions requires selecting individual edges or refitting network-specific models.

{}

\FloatBarrier

\section{Discussion}
\label{sec:discussion}

The central finding is that the covariance organization associated with age across people is not, especially in early and middle adulthood, the same as the covariance trajectory followed within individuals. This distinction is more than a technical separation of two regression effects. It changes the biological interpretation of lifespan connectivity. Population age signatures can be informative markers of age while still failing to describe the direction in which an individual's functional organization changes over time. By estimating the two directions separately, the analysis makes their agreement or divergence an empirical question rather than an assumption.

The DLBS results show that this divergence is structured across the lifespan. The within-person and between-person covariance directions are negatively aligned through approximately age 60, after which the discrepancy weakens. At the same time, the magnitude of longitudinal covariance reorganization is greatest around midlife rather than increasing monotonically into old age. These findings suggest that functional brain aging is not well summarized by a single slope or by a fixed network signature. The trajectory changes in both magnitude and orientation, with the most pronounced reorganization occurring at different stages of adulthood rather than simply at the oldest ages.

The low effective dimension of the systematic longitudinal signal adds an important complementary result. Only a few coordinated directions are needed to describe reproducible covariance aging, even though covariance variation among individual scans is much richer. This concentration does not imply a fixed low-dimensional aging pattern. Nearby ages have highly similar operators, but similarity decreases gradually over wider age separations, indicating smooth reorientation across adulthood. The biological picture is therefore one of a parsimonious but evolving process: aging is organized, yet the organization through which it is expressed is not static.

Network localization shows why the whole-operator perspective matters. Roughly three quarters of off-diagonal longitudinal aging energy is carried by covariance between functional systems. The disagreement between longitudinal and cross-sectional aging is especially visual-centered in early and middle adulthood, involving relationships of the visual system with default, control, dorsal-attention, and limbic systems. Much of the global discrepancy attenuates later in life, but visual--limbic covariance remains strongly discordant even at age 90. A global summary alone would miss this persistence. Conversely, an edgewise analysis would fragment a coordinated system-level pattern into many separate tests. The operator decomposition preserves the global biological question while showing where the disagreement is expressed.

The model connects these findings to uncertainty in the finite BOLD runs while allowing the directions of covariance aging to be learned from longitudinal observations.

The results also refine how cross-sectional studies of brain aging should be interpreted. Cross-sectional signatures remain valuable for describing population differences and for identifying age-associated organization. The present analysis does not invalidate those findings. Instead, it shows that their biological meaning should not automatically be extended to within-person change. The distinction is most consequential when the scientific aim is mechanistic---to understand how functional organization changes as a person ages---or when cross-sectional patterns are used to motivate trajectories, intervention windows, or markers of progression.

The proposed framework also has applications beyond functional neuroimaging. Across biomedical research, repeated measurements often capture complex relationships among biological signals that evolve over time. Similar challenges arise in electrophysiology, physiological monitoring, wearable sensor data, and molecular studies, where the goal is to understand how patterns of dependence change within individuals. In these settings, covariance must be estimated from limited observations, and changes observed within an individual may differ from associations found across a population. By accounting for these differences, our approach offers a way to study how biological systems reorganize over time while preserving the underlying dependence structure.




\section*{Acknowledgments}

The authors thank the investigators and participants of the Dallas Lifespan Brain Study for making this longitudinal resource publicly available.

\printbibliography

\clearpage
\begingroup
\onehalfspacing
\setcounter{section}{0}
\setcounter{subsection}{0}
\setcounter{figure}{0}
\setcounter{table}{0}
\setcounter{theorem}{0}
\renewcommand{\thesection}{S\arabic{section}}
\renewcommand{\thesubsection}{\thesection.\arabic{subsection}}
\renewcommand{\thefigure}{S\arabic{figure}}
\renewcommand{\thetable}{S\arabic{table}}
\renewcommand{\thetheorem}{S\arabic{theorem}}
\renewcommand{\theHsection}{supp.\arabic{section}}
\renewcommand{\theHfigure}{supp.\arabic{figure}}
\renewcommand{\theHtable}{supp.\arabic{table}}
\renewcommand{\theHtheorem}{supp.\arabic{theorem}}
\begin{center}
{\Large\bfseries Supplementary Materials for\\[5pt] Longitudinal Brain Aging in Covariance Space}\par
\vspace{12pt}
Montserrat Fuentes and Yi Zhao
\end{center}
\vspace{1em}
\section{Proof of the positive-definiteness theorem}
The main paper models covariance through a matrix exponential. This section establishes why the resulting matrix is a valid covariance matrix at every age. It also explains why smooth changes in the model produce a smooth covariance trajectory. The result provides the mathematical foundation for interpreting within-person covariance change.

\begin{theorem}[Positive-definite covariance trajectory]
Let $\bm\Sigma_{\star}$ be a positive-definite reference covariance matrix, let $\bm D$ be a positive diagonal matrix, and let $\bm A$ be a real symmetric matrix. Then
\[
\bm\Sigma=\bm\Sigma_{\star}^{1/2}\bm D^{1/2}
\exp(\bm A)\bm D^{1/2}\bm\Sigma_{\star}^{1/2}
\]
is positive definite. If the matrices defining this covariance vary smoothly with age, the resulting covariance trajectory is smooth.
\end{theorem}

\begin{proof}
Because $\bm A$ is symmetric, it has an orthogonal eigendecomposition $\bm A=\bm V\bm\Lambda\bm V^{\top}$. Its matrix exponential is $\exp(\bm A)=\bm V\operatorname{diag}(e^{\lambda_1},\ldots,e^{\lambda_p})\bm V^{\top}$. Every eigenvalue of this exponential is strictly positive, so $\exp(\bm A)$ is positive definite.

Write $\bm C=\bm\Sigma_{\star}^{1/2}\bm D^{1/2}$. Both factors are invertible. For any nonzero vector $\bm x$, the vector $\bm C^{\top}\bm x$ is also nonzero, and hence
\[
\bm x^{\top}\bm\Sigma\bm x
=(\bm C^{\top}\bm x)^{\top}\exp(\bm A)
 (\bm C^{\top}\bm x)>0.
\]
Thus $\bm\Sigma$ is positive definite. The matrix exponential is a smooth function of its matrix argument. Matrix square roots and products are also smooth on the positive-definite domain. Consequently, a smoothly varying symmetric operator and smoothly varying positive-definite scaling factors define a smooth covariance trajectory.
\end{proof}

The theorem concerns the covariance matrix itself rather than a particular coordinate representation of its low-dimensional operators. Changing the basis of a learned span while transforming its coefficient matrix accordingly leaves the induced operator unchanged. This is a representation identity and does not assert that a coordinate-dependent shrinkage prior is rotation invariant.

\clearpage
\section{Technical definitions and justification}
\label{sec:supptech}

The following details clarify how run-level uncertainty enters the model and how the fitted covariance trajectory is interpreted. They use the notation of the main paper and do not introduce an additional model.

\subsection{Information in a resting-state run}
For run $r$ at visit $j$ of participant $i$, let $T_{ijr}$ be the number of retained BOLD frames and let $\bm R_{ijr}$ denote their temporal correlation matrix, evaluated at the original frame positions after censoring. The effective covariance information used in the analysis is
\[
\nu_{ijr}=\frac{T_{ijr}^{2}}{\operatorname{tr}(\bm R_{ijr}^{2})},
\qquad \nu_{ijr}^{\ast}=\frac{\nu_{ijr}}{\phi}.
\]
The first expression reduces information when temporally correlated frames carry redundant information. The second accounts for additional dispersion observed between runs at the same visit. In particular, the effective information is not interpreted as the number of independent BOLD frames.

For a run sample covariance $\bm S_{ijr}$ and its modeled covariance $\bm\Sigma_{ijr}$, the parameter-dependent part of the working covariance log-likelihood is proportional to
\[
-\frac{1}{2}\sum_{ijr}\nu_{ijr}^{\ast}
\left\{\log|\bm\Sigma_{ijr}|+
\operatorname{tr}(\bm\Sigma_{ijr}^{-1}\bm S_{ijr})\right\}.
\]
This expression is a quasi-likelihood used for generalized Bayesian updating. It does not claim that each sample covariance follows an exact Wishart distribution, particularly when effective information is below the number of parcels. The estimated dispersion factor changes the scale of run contributions and therefore the posterior uncertainty.

\subsection{Covariance-scale derivatives}
Let $\bm A(a)$ denote a differentiable symmetric tangent-state trajectory, and let $\bm C(a)$ collect the positive-definite scaling factors so that $\bm\Sigma(a)=\bm C(a)\exp\{\bm A(a)\}\bm C(a)^\top$. Its derivative is
\[
\begin{aligned}
\frac{d\bm\Sigma(a)}{da}
={}&\bm C'(a)\exp\{\bm A(a)\}\bm C(a)^\top\\
&+\bm C(a)\,\mathcal L_{\exp}\!\left(\bm A(a),\bm A'(a)\right)\bm C(a)^\top\\
&+\bm C(a)\exp\{\bm A(a)\}\bm C'(a)^\top,
\end{aligned}
\]
where the Fr\'echet derivative of the matrix exponential is
\[
\mathcal L_{\exp}(\bm A,\bm E)
=\int_0^1\exp\{(1-t)\bm A\}\,\bm E\,\exp(t\bm A)\,dt.
\]
The longitudinal derivative $\bm D_W(a)$ is obtained by differentiating with respect to the within-person age component. The between-person derivative $\bm D_B(a)$ is defined analogously. The corresponding reorganization rate is $R_W(a)=\|\bm D_W(a)\|_F$, and their alignment is the Frobenius cosine reported in the main paper. These are covariance-scale quantities, not derivatives of individual correlations.

\subsection{Learned directions and network summaries}
When a symmetric operator is represented in an orthonormal span $\bm Q$ as $\bm A=\bm Q\bm H\bm Q^\top$, with $\bm Q^\top\bm Q=\bm I$, its exponential can be evaluated without a full eigendecomposition of the ambient covariance dimension:
\[
\exp(\bm Q\bm H\bm Q^\top)
=\bm I+\bm Q\{\exp(\bm H)-\bm I\}\bm Q^\top.
\]
This identity follows because the operator is zero on the orthogonal complement of the span. It is a computational identity, not a separate assumption about the covariance trajectory. An orthogonal change of coordinates within the span leaves the represented operator unchanged when the coefficient matrix is transformed accordingly. This does not require the coordinate-level prior to be invariant.

For the seven functional systems, let $\mathcal I_g$ denote the parcels assigned to network $g$. The block of the longitudinal derivative indexed by $\mathcal I_g\times\mathcal I_h$ describes covariance change between those systems. Its squared Frobenius norm measures the contribution of that block to the overall change. Summing squared entries with parcel indices in different systems, and dividing by the sum over all off-diagonal entries, gives the between-system energy fraction. Network-pair alignment is the Frobenius cosine of the corresponding blocks of $\bm D_W(a)$ and $\bm D_B(a)$ when both block norms are nonzero. These summaries decompose the same fitted derivatives and do not require fitting separate models.

\clearpage
\section{Simulation study}
\label{sec:suppsim}

The simulation examines whether the model can recover the direction of covariance change experienced within individuals. It also tests what happens when the population age gradient is used in its place. Because each visit is informed by repeated runs, we assess how accounting for variation between runs affects interval coverage. The experiment focuses on these features of the model rather than reproducing every aspect of the DLBS data.

We generated 200 data sets with $p=16$ dimensions, 120 participants, three visits per participant, and two runs per visit. Participant mean ages were distributed across adulthood, with visits separated by approximately four years. The latent tangent state was
\[
\bm A_{ij}
=
w_{ij}\bm B_W(\bar a_i)+
b_i\bm B_B,
\]
where $b_i$ denotes the centered between-person age term, $\bm B_W(a)$ varied smoothly across age through four radial basis coefficients, and $\bm B_B$ represented a distinct between-person age direction. The longitudinal operator had approximately three dominant directions and changed orientation with age. The aging directions partially overlapped the leading eigenvectors of the baseline covariance. This allowed us to examine whether choosing directions from pooled variation would miss some of the aging signal. The latent covariance was generated by the same exponential geometry used in the model,
\[
\bm\Sigma_{ij}
=
\bm L_0\exp(\bm A_{ij})\bm L_0^\top .
\]
Run-level tangent perturbations had variance proportional to $\phi/\nu_{ijr}$, with $\phi=3.5$ and mean run information near 62. The overdispersion factor changes the uncertainty assigned to the runs without changing their relative weights. It therefore affects the intervals but not the point estimates.

For comparison, we constructed intervals that ignore additional run-to-run dispersion by setting $\phi=1$. These intervals share the calibrated point estimates, allowing us to examine uncertainty calibration directly. We also evaluated the between-person age gradient as a substitute for the longitudinal direction. To assess the effect of choosing directions from pooled variation, we calculated how much of the true aging signal was retained by the five leading baseline covariance eigenvectors.

The proposed longitudinal estimator recovered the covariance-aging direction with high accuracy throughout adulthood (Table~\ref{tab:simulation-main}). Mean cosine similarity with the true $\bm D_W(a)$ was 0.9995 at ages 30 and 50, 0.9993 at age 70, and 0.9949 at age 90. RMSE for the global reorganization rate remained below $7.4\times10^{-4}$ at all ages, and RMSE for $\rho_{WB}(a)$ ranged from 0.0033 to 0.0099. In contrast, treating the cross-sectional gradient as the longitudinal direction produced negative operator similarities at every evaluated age, from $-0.540$ at age 30 to $-0.310$ at age 90 (Figure~\ref{fig:sim-operator}). The difference is therefore not simply one of magnitude. The population age gradient can point in a different multivariate direction.

\begin{figure}[t]
\centering
\includegraphics[width=.72\textwidth]{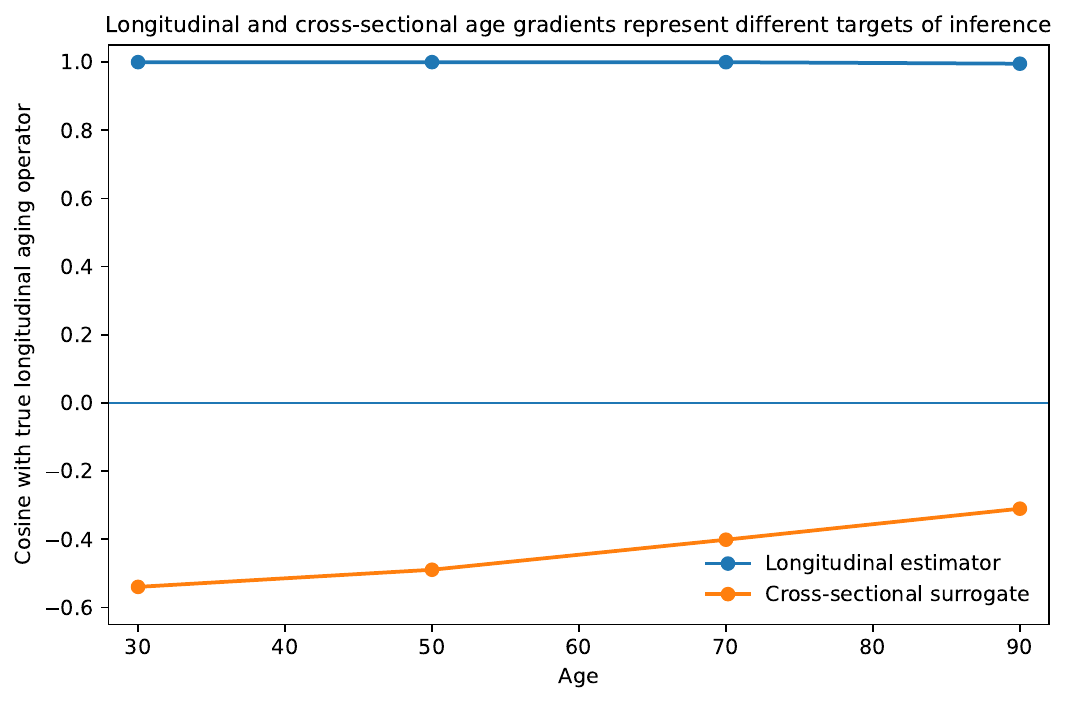}
\caption{Simulation recovery of the longitudinal covariance-aging direction. The proposed longitudinal estimator is nearly collinear with the true within-person derivative across adulthood, whereas using the cross-sectional age gradient as a surrogate yields a negatively aligned direction. Results are averages over 200 Monte Carlo replicates.}
\label{fig:sim-operator}
\end{figure}

Repeated-run calibration primarily affected uncertainty, as expected theoretically. For $\rho_{WB}(a)$, the calibrated 95\% intervals had empirical coverage of 0.970, 0.950, 0.950, and 0.940 at ages 30, 50, 70, and 90. Ignoring $\phi$ reduced coverage to 0.720, 0.735, 0.710, and 0.750 (Figure~\ref{fig:sim-coverage}). Coverage for $R_W(a)$ likewise improved from 0.650--0.700 under $\phi=1$ to 0.905--0.930 through age 70. At age 90 the calibrated coverage for $R_W(a)$ was 0.835, indicating residual finite-sample undercoverage for the positive norm functional when the longitudinal signal is weaker. We therefore use the simulation to support calibrated uncertainty propagation while retaining caution about norm-based interval interpretation at the oldest ages.

\begin{figure}[t]
\centering
\includegraphics[width=.72\textwidth]{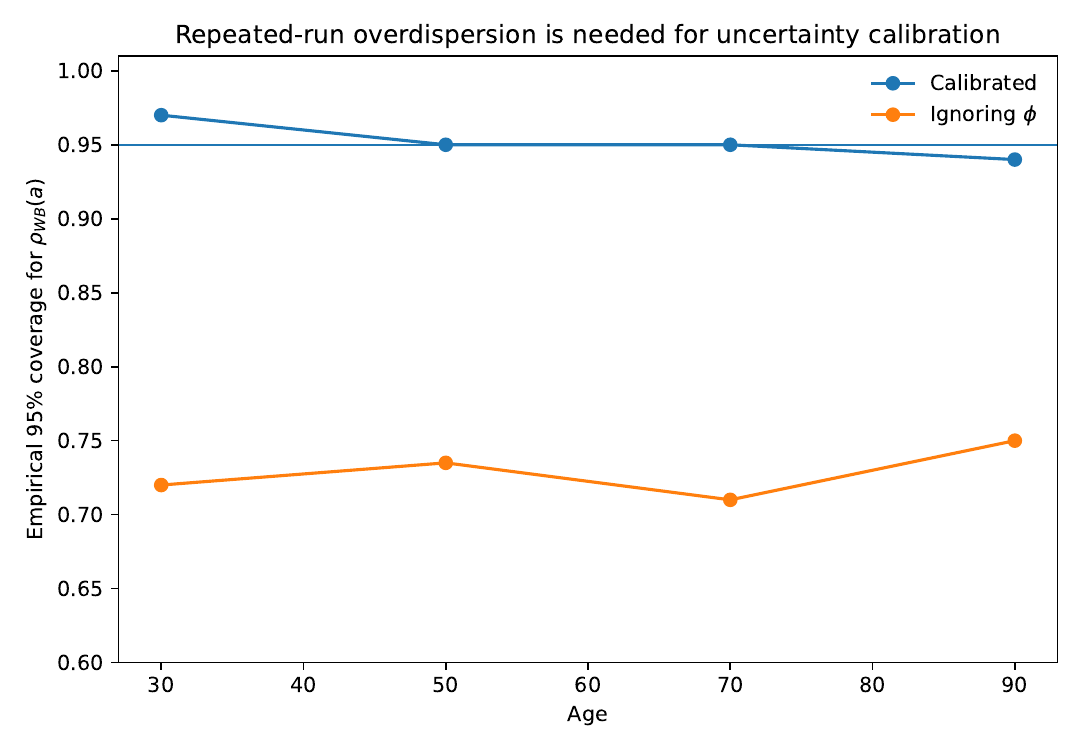}
\caption{Empirical coverage of nominal 95\% intervals for the longitudinal/between-person alignment $\rho_{WB}(a)$. Propagating the repeated-run overdispersion factor produces coverage close to nominal, whereas setting $\phi=1$ yields substantial undercoverage.}
\label{fig:sim-coverage}
\end{figure}

The pooled-variance comparison addresses a different concern. The five leading baseline covariance eigenvectors retained only 47.3\%, 46.4\%, 43.6\%, and 41.1\% of the true longitudinal aging-operator energy at ages 30, 50, 70, and 90 (Figure~\ref{fig:sim-pca}). This does not imply that principal components are generally unsuitable. It shows that the directions accounting for the greatest overall covariance variation can differ from those that capture reproducible within-person aging. This motivates learning the aging-associated tangent subspace directly from the longitudinal structure rather than determining it in advance from pooled variance.

\begin{figure}[t]
\centering
\includegraphics[width=.72\textwidth]{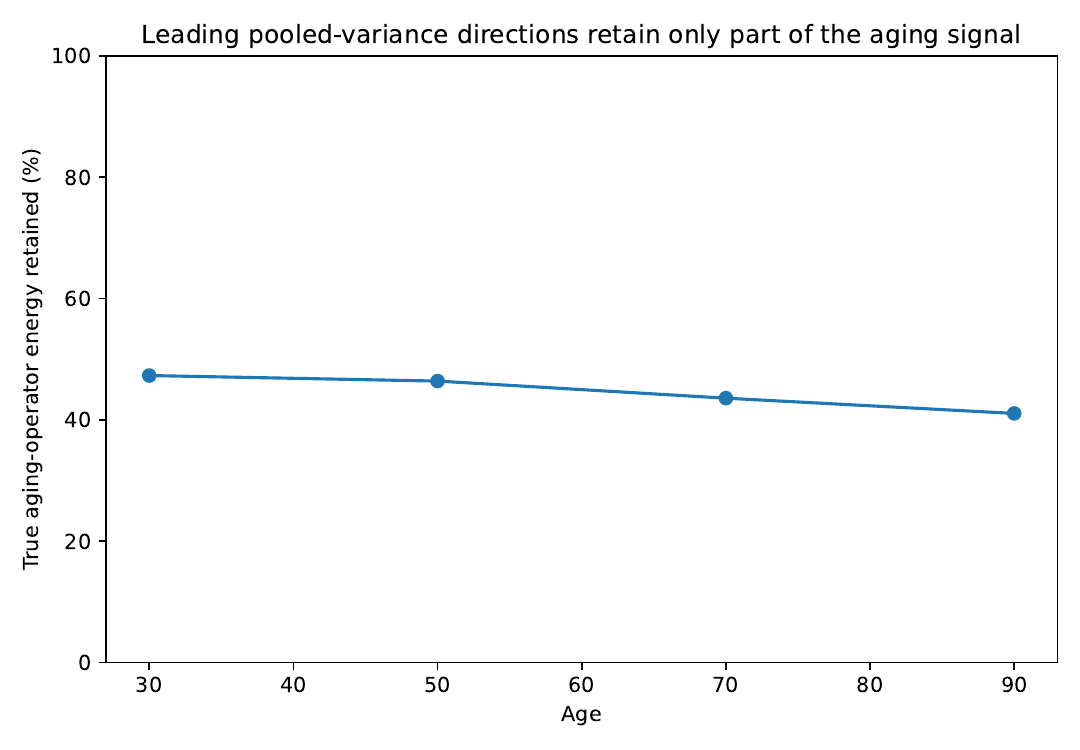}
\caption{Fraction of the true longitudinal aging-operator energy retained by the five leading eigenvectors of the baseline covariance. The aging subspace was constructed to overlap partially, but not completely, with dominant pooled-variance directions. Even in this nondegenerate setting, variance-based reduction discards more than half of the longitudinal signal.}
\label{fig:sim-pca}
\end{figure}

\begin{table}[t]
\centering
\caption{Simulation performance over 200 Monte Carlo replicates. Operator cosine is the cosine similarity between the estimated and true within-person covariance-aging derivatives. Coverage refers to nominal 95\% intervals. The calibrated and $\phi=1$ analyses have the same point estimator because the global overdispersion factor rescales uncertainty rather than the relative run weights.}
\label{tab:simulation-main}
\small
\resizebox{\textwidth}{!}{%
\begin{tabular}{llrrrrr}
\toprule
Method & Age & $R_W$ RMSE & $\rho_{WB}$ RMSE & Operator cosine & $R_W$ coverage & $\rho_{WB}$ coverage\\
\midrule
Calibrated & 30 & 0.00035 & 0.0033 & 0.9995 & 0.920 & 0.970\\
Calibrated & 50 & 0.00026 & 0.0035 & 0.9995 & 0.930 & 0.950\\
Calibrated & 70 & 0.00029 & 0.0042 & 0.9993 & 0.905 & 0.950\\
Calibrated & 90 & 0.00074 & 0.0099 & 0.9949 & 0.835 & 0.940\\
\addlinespace
Ignore $\phi$ & 30 & 0.00035 & 0.0033 & 0.9995 & 0.685 & 0.720\\
Ignore $\phi$ & 50 & 0.00026 & 0.0035 & 0.9995 & 0.700 & 0.735\\
Ignore $\phi$ & 70 & 0.00029 & 0.0042 & 0.9993 & 0.670 & 0.710\\
Ignore $\phi$ & 90 & 0.00074 & 0.0099 & 0.9949 & 0.650 & 0.750\\
\addlinespace
Cross-sectional surrogate & 30 & 0.08415 & 1.5419 & -0.5397 & -- & --\\
Cross-sectional surrogate & 50 & 0.06868 & 1.4916 & -0.4894 & -- & --\\
Cross-sectional surrogate & 70 & 0.05676 & 1.4037 & -0.4013 & -- & --\\
Cross-sectional surrogate & 90 & 0.05586 & 1.3132 & -0.3103 & -- & --\\
\bottomrule
\end{tabular}%
}
\end{table}

The simulation study provides strong support for the proposed methodology by demonstrating its ability to recover longitudinal patterns of covariance change and distinguish them from population-level age differences. The results also highlight the importance of modeling within-person covariance dynamics rather than relying on pooled-variance directions. Together, these findings support the validity of the proposed approach and its application to understanding longitudinal brain reorganization in the DLBS data.



\endgroup

\end{document}